\documentclass[10pt,final, twocolumn]{IEEEtran}
\usepackage{amsmath,amssymb,amsfonts}
\usepackage{graphicx,color}
\usepackage{textcomp}
\usepackage{xcolor}
\usepackage{hyperref}
\hypersetup{hidelinks=true}
\usepackage[linesnumbered,ruled,vlined,noend]{algorithm2e}

\usepackage{amssymb,amsthm}
\usepackage{epsfig,verbatim}
\usepackage{caption}
\usepackage{subcaption}
\usepackage{setspace}
\usepackage[flushleft]{threeparttable}
\usepackage{color}
\usepackage{epstopdf}
\usepackage{accents}
\usepackage{acronym}
\usepackage{booktabs}
\usepackage{mathtools} 
\usepackage{enumitem}
\usepackage{multirow}
\usepackage{dblfloatfix}
\usepackage{bbm}
\AtBeginDocument{\definecolor{tmlcncolor}{cmyk}{0.93,0.59,0.15,0.02}\definecolor{NavyBlue}{RGB}{0,86,125}}

\newcommand{\myVec}[1]{{\boldsymbol{#1}}}

\newcommand{\mySet}[1]{\mathcal{#1}}
	
\newcommand{\mys}{{\myVec{s}}}	

\newcommand{\Weights}{\myVec{\theta}}
\newcommand{\WeightsSet}{\myVec{\Theta}}
\newcommand{\Network}{\myVec{\varphi}}
\DeclareMathOperator*{\argmin}{argmin}
\makeatletter
\newcommand{\removelatexerror}{\let\@latex@error\@gobble}
\makeatother

\newcommand{\Blklen}{B}			 			
\newcommand{\Blkset}{\mySet{\Blklen}}

\newcommand{\figheight}{0.198\textheight}
\newcommand{\figwidth}{0.9\textwidth}



\newtheorem{corollary}{Corollary}
\newtheorem{proposition}{Proposition}

\acrodef{fc}[FC]{fully-connected}
\acrodef{adc}[ADC]{analog-to-digital convertor}
\acrodef{cs}[CS]{compressed sensing}
\acrodef{dtft}[DTFT]{discrete-time Fourier transform}
\acrodef{dnn}[DNN]{deep neural network} 
\acrodef{csi}[CSI]{channel state information}
\acrodef{bpsk}[BPSK]{binary phase shift keying}
\acrodef{map}[MAP]{maximum a-posteriori probability}
\acrodef{snr}[SNR]{signal-to-noise ratio}
\acrodef{bs}[BS]{base station} 
\acrodef{iot}[IOT]{Interent of Things}
\acrodef{mimo}[MIMO]{multiple-input multiple-output}
\acrodef{siso}[SISO]{single-input single-output}
\acrodef{mse}[MSE]{mean-squared error}
\acrodef{pdf}[PDF]{probability density function}
\acrodef{rv}[RV]{random variable}
\acrodef{ml}[ML]{machine learning}
\acrodef{mc}[MC]{monte carlo}
\acrodef{fec}[FEC]{forward error correction}
\acrodef{qpsk}[QPSK]{quadrature phase shift keying}
\acrodef{rs}[RS]{Reed-Solomon}
\acrodef{lti}[LTI]{linear time-invariant}
\acrodef{wss}[WSS]{wide-sense stationary}
\acrodef{psd}[PSD]{power spectral density}
\acrodef{ser}[SER]{symbol error rate} 
\acrodef{ber}[BER]{bit error rate} 
\acrodef{gd}[GD]{gradient descent}
\acrodef{sgd}[SGD]{stochastic gradient descent} 
\acrodef{isi}[ISI]{intersymbol interference}  
\acrodef{awgn}[AWGN]{additive zero-mean white Gaussian noise} 
\acrodef{ut}[UT]{user terminal} 
\acrodef{mmw}[mmWave]{millimeter wave}
\acrodef{noma}[NOMA]{non-orthogonal multiple access}
\acrodef{mac}[MAC]{mulitple access channel}
\acrodef{fl}[FL]{Federated learning}
\acrodef{lstm}[LSTM]{long short-term memory}
\acrodef{maml}[MAML]{model-agnostic meta-learning}
\acrodef{sic}[SIC]{soft interference cancellation}
\acrodef{pmf}[PMF]{probability mass function}
\acrodef{crc}[CRC]{cyclic redundancy check}
\acrodef{ece}[ECE]{expected calibration error}
\acrodef{lbd}[LBD]{learnable Bernoulli dropout}
\acrodef{arm}[ARM]{Augment-REINFORCE-Merge}
\acrodef{kl}[KL]{Kullback-Leibler}
\acrodef{ai}[AI]{artificial intelligence}
\acrodef{ddm}[DDM]{Drift Detection Method}
\acrodef{pht}[PHT]{Page Hinkley Test}
\acrodef{rnn}[RNN]{recurrent neural network}
\acrodef{cnn}[CNN]{convolutional neural network}

\begin{document}

\markboth{}{Author {et al.}}

\title{Modular Hypernetworks for Scalable and Adaptive Deep MIMO
Receivers}

\author{  
	\IEEEauthorblockN{Tomer Raviv and  Nir Shlezinger
	} 
	\thanks{
		T. Raviv and N. Shlezinger are with the School of ECE, Ben-Gurion University of the Negev, Beer-Sheva, Israel (e-mail: tomerraviv95@gmail.com; nirshl@bgu.ac.il).}

	
}

\maketitle

\begin{abstract}
\Acp{dnn} were shown to facilitate the operation of uplink \ac{mimo} receivers, with emerging architectures augmenting modules of classic receiver processing. Current designs consider static \acp{dnn}, whose architecture is fixed and weights are pre-trained. This induces a notable challenge, as the resulting \ac{mimo} receiver is suitable for a given configuration, i.e., channel distribution and number of users, while in practice these parameters change frequently with network variations and users leaving and joining the network. 
%
In this work, we tackle this core challenge of \ac{dnn}-aided \ac{mimo} receivers. We build upon the concept of {\em hypernetworks}, augmenting the receiver with a pre-trained deep model whose purpose is to update the weights of the \ac{dnn}-aided receiver upon instantaneous channel variations. We design our hypernetwork to augment {\em modular} deep receivers, leveraging their modularity to have the hypernetwork adapt not only the weights, but also the architecture. Our  modular hypernetwork leads to a \ac{dnn}-aided receiver whose architecture and resulting complexity adapts to the number of users, in addition to channel variations, without  re-training.  Our numerical studies demonstrate superior error-rate performance of modular hypernetworks in time-varying channels compared to static pre-trained receivers, while providing rapid adaptivity and scalability to network variations.   
\end{abstract}


\section{Introduction}
Deep learning is envisioned to play a key role in enabling  future wireless communication systems to meet their constantly growing demands~\cite{dai2020deep}. A core aspect which is expected to greatly benefit from proper augmentation of deep learning techniques is receiver design~\cite{wang2020artificial}.  \acp{dnn} can facilitate coping with both {\em model-deficiency}~\cite{simeone2022machine}, i.e., complex channel models, as well as {\em algorithm-deficiency}~\cite{kee2024review},  where classic receiver processing is not suitable.

Despite their potential, deployment of \ac{dnn}-aided receivers, termed {\em deep receivers}, is subject to  several  challenges, arising from the fundamental differences between wireless communication 
 and traditional deep learning domains (such as vision and natural language processing)~\cite{wang2020artificial,raviv2024adaptive}. 
These challenges include the {\em dynamic nature} of wireless channels~\cite{wang2020artificial,aoudia2021end}, and the {\em limited compute/power resources} of wireless devices~\cite{raviv2024adaptive,zhu2020toward}. The dynamic nature of wireless channels implies that the receiver task, dictated by the data distribution and the number of received messages,  changes rapidly in time. This evolution can occur either on the physical level, e.g., variations in the \ac{snr} and channel transfer function, or  the network level, e.g., the number of users transmitting  in the uplink.

A common approach to deal with dynamic channels, coined {\em joint learning}~\cite{oshea2017introduction}, trains the \ac{dnn} over a wide range of channel conditions. Once trained, the \ac{dnn}-aided receiver is deployed statically~\cite{xia2020note}, namely, its weights and architecture do not vary.  This form of learning seeks a non-coherent receiver at the cost of performance degradation as compared with coherent operation~\cite{raviv2024adaptive}, and cannot cope with variations in the network, as its number of inputs and outputs is fixed. 
While static receivers can operate under some forms of channel variations by providing an  estimate of the channel parameters as additional features~\cite{honkala2021deeprx}, the  architecture is still fixed, and cannot adapt to network variations. 

An alternative approach aims at providing increased flexibility by re-training the \ac{dnn} on device. Such {\em online learning} uses pilots~\cite{shlezinger2019viterbinet,shlezinger2019deepsic}, data augmentation~\cite{raviv2022data} and locally decoded messages~\cite{schibisch2018online} to repeatedly adapt the weights of the receiver upon channel variations. While online learning yields increased flexibility to dynamic channels, it induces notable excessive complexity due to its frequent re-training procedures. 
Existing approaches to facilitate online learning include:  
$(i)$ Designing deep receiver architectures as a form of model-based deep learning~\cite{shlezinger2022model, shlezinger2023model}, using classic receiver processing as an informative and parameters-compact inductive bias~\cite{balatsoukas2019deep, zappone2019wireless}, that makes them more amenable to rapid adaptation compared to black-box \acp{dnn}~\cite{raviv2024adaptive};
$(ii)$ Altering learning algorithms by, e.g., incorporating meta-learning to achieve fast training~\cite{raviv2022online, park2020learning}, or combining Bayesian learning to mitigate overfitting on scarce data~\cite{zecchin2023robust, raviv2023modular}; 
$(iii)$ Reducing online training frequency by, e.g., using dedicated mechanisms to detect when to re-train and which module requires adaptation~\cite{uzlaner2024concept}. 

Despite these recent advances in online learning, its excessive complexity makes it  challenging to implement  due to the limited resources of wireless devices~\cite{raviv2024adaptive,zhu2020toward}. Moreover, the aforementioned training approaches have primarily focused on adjusting a {\em fixed architecture}, and are not designed to handle settings necessitating {\em changes in the architecture}.

From a deep learning perspective, adaptivity of \acp{dnn} can be achieved using   {\em hypernetworks}~\cite{ha2016hypernetworks}.
Hypernetworks are \acp{dnn} whose outputs are utilized as the weights for another primary \ac{dnn}. They have been demonstrated in \cite{galanti2020modularity} to achieve low generalization error accross different data distributions with only a small additional complexity during inference. In  wireless communications, hypernetworks have primarily been considered to adapt receivers  for physical layer variations~\cite{goutay2020deep,liu2024hypernetwork,nachmani2019hyper,pratik2021neural}. For instance, in \cite{liu2024hypernetwork}, a low-complexity hypernetwork was employed to adjust the parameters of a black-box channel predictor \ac{dnn}, with the goal of handling non-stationary wireless channels. For detection, the work  \cite{goutay2020deep} proposed using hypernetworks to mitigate the necessity of retraining a \ac{dnn}-based \ac{mimo} receiver for each channel realization. However, these studies have only considered changes in the channel that do not require {\em architectural adaptation}, which is extremely challenging and complex to realize using black-box \acp{dnn}. The usefulness of model-based deep learning for deep receivers, combined with  its modularity and the ability to associate internal modules with different users, motivate integrating hypernetworks with  modular architectures to yield deep receivers that  rapidly adapt to variations in both the channel and the network without  online training. 


In this work we address this gap and propose an approach that  adapts \ac{dnn}-aided receivers to variations in both the channel and the network without necessitating additional re-training. We do so by fusing the modularity of model-based deep learning receiver architectures, combined with a modular hypernetwork that generates during inference the {\em weights} for {a \em varying number of modules} of a deep receiver.
In doing so, we preserve the low complexity overhead of static model-based deep learning, while enjoying some level of adaptivity as offered by online learning, without its excessive compute and limitation to a fixed architecture.

We focus on uplink \ac{mimo} receivers, considering the modular \ac{dnn}-aided DeepSIC architecture of~\cite{shlezinger2019deepsic}, which is comprised of a set of modules corresponding to different users, as our main running deep receiver. For such  modular architectures, we develop a lightweight hypernetwork framework that learns  to generate {\em complete sub-modules} and weights of the architecture. The resulting modular hypernetwork allows the deep receiver  to accommodate an arbitrary number of users with different time-varying channel conditions, while providing  elastic inference complexity, as its number of \acp{dnn} modules is adapted to match the varying number of users. We introduce a dedicated training method that can be carried out offline based on channel measurements and simulations, without necessitating on-device adaptation. 
We numerically demonstrate the effectiveness of our modular hypernetwork, showing that it allows the deep receiver not only to outperform joint training in time-varying channels, but also approach the \ac{ser} of computationally intensive online training, while maintaining a significantly lower complexity, as shown analytically.


The rest of this paper is organized as follows: 
Section~\ref{sec:background} reviews the system model. Section~\ref{sec:method} presents our modular hypernetwork-based framework  and analyzes its complexity. Section~\ref{sec:Simulation} numerically evaluates the suggested framework, and  Section~\ref{sec:conclusion} concludes the paper.

\section{System Model}
\label{sec:background}


\subsection{Communication System}
\label{subsec:system_model}
We consider an uplink block-fading \ac{mimo} digital communication system. The system supports up to $K^{\rm max}$ single-antenna users, that are transmitting  symbols to a base station equipped with $N$ antennas, where $N\geq K^{\rm max}$. 
Each block is comprised of $\Blklen^{\rm tran}$ time instances, during which the channel parameters, denoted $\myVec{H}[t]\sim \mathcal{H}$ for the $t$th block (where $\mathcal{H}$ is the distribution of channel coefficients), are constant. 
We allow the number of users to change between blocks, corresponding to users joining and leaving the network. Accordingly, we denote the number of users in the block with index $t$ by $K[t] \in \{1,\ldots, K^{\max}\}$. Specifically, at each $t$th block, $K[t]$ users simultaneously transmit $\Blklen^{\rm tran}$ independent messages denoted $\myVec{s}_i[t] = \big[\myVec{s}_i^{(1)}[t],\ldots,\myVec{s}_i^{(K[t])}[t]\big]\in\mySet{S}^{K[t]}$, where $i \in \Blkset\triangleq \{1,\ldots, \Blklen^{\rm tran}\}$ is the symbol index within the block,  and $\mySet{S}$ is the set of constellation points. The transmitted symbols block $\mys ^{\rm tran}[t]:= \{\myVec{s}_i[t]\}_{i\in \Blkset}$ is divided into $\Blklen^{\rm pilot}$ pilots that are known to the receiver and appear first,  denoted $\myVec{s}^{\rm pilot}[t]$, and  $\Blklen^{\rm info} = \Blklen^{\rm tran}-\Blklen^{\rm pilot}$ information symbols, denoted $\myVec{s}^{\rm info}[t]$, that contain the digital message.

To accommodate complex and non-linear channel models, we represent the channel mapping by a generic memoryless conditional distribution. Accordingly, the corresponding received signal vector $\myVec{y}_i[t] \in \mathbb{C}^N$ is determined as
\begin{equation}
\label{eq:general_channel_mapping}
\myVec{y}_i[t] \sim P_{\myVec{H}[t]}(\myVec{y}_i[t]|\myVec{s}_i[t]),
\end{equation}
which is subject to the unknown conditional  distribution $P_{\myVec{H}[t]}(\cdot|\cdot)$  that depends on the current channel parameters $\myVec{H}[t]$. We denote the corresponding observations of $\myVec{s}^{\rm pilot}[t]$ as $\myVec{y}^{\rm pilot}[t]$, and of $\myVec{s}^{\rm info}[t]$ as $\myVec{y}^{\rm info}[t]$.

The received channel outputs in \eqref{eq:general_channel_mapping} are processed by the receiver. We denote the mapping function of the receiver applied to any of the symbols during the $t$th block as $\mathcal{F}[t]:\mathbb{C}^N \rightarrow \mySet{S}^{K[t]}$, whose goal is to recover correctly each digital symbol in $\myVec{s}^{\rm info}_i[t]$ from $\myVec{y}^{\rm info}_i[t]$. The number of users in each block, $K[t]$, is assumed to be known to the receiver.


\subsection{Problem Formulation}
\label{subsec:problem}
Let us denote the estimated message for all $K[t]$ users in a given block index $t$ as $\hat{\myVec{s}}_i^{\rm info}[t]:= \mathcal{F}(\myVec{y}_i^{\rm info}[t])$, and the average computational complexity associated with the estimation of the $\Blklen^{\rm info}$ symbols as $C_{\rm avg}$. Our objective is to minimize the \ac{ser} over the information symbols, while reducing the computational overhead. Over a horizon of $T$ blocks, this is formulated as
\begin{align*}
\label{eqn:Problem} 
    	&\min \frac{1}{T\cdot B^{\rm info}}\sum_{t=1}^T\sum_{i=B^{\rm pilot}+1}^{B^{\rm tran}} \Pr\left( \hat{\myVec{s}}^{\rm info}_i[t] \neq  \myVec{s}^{\rm info}_i[t] \right),
\end{align*}
while keeping $C_{\rm avg}$ as low as possible. 

\subsection{Deep Receivers}
\label{subsec:DeepRx}
We focus on \ac{mimo} receivers implemented using \acp{dnn}, i.e.,  deep receivers. Such receivers are parameterized at time $t$ by the weights vector $\WeightsSet[t]$, and can be trained to operate in complex and unknown channel models as in \eqref{eq:general_channel_mapping}. Accordingly, we write the receiver processing as 
\begin{equation}
\label{eqn:RxMapping}
\mathcal{F}_{\WeightsSet[t]}[t]:\mathbb{C}^N \rightarrow \mySet{S}^{K[t]}, 
\end{equation}
and consider soft-output receivers that output a conditional  distribution over $\mySet{S}^{K[t]}$ denoted $P_{\WeightsSet[t]}$, such that
\begin{equation}
    \mathcal{F}_{\WeightsSet[t]}[t]\left(\myVec{y}\right) = \mathop{\arg \max}_{\myVec{s} \in \mySet{S}^{K[t]}} P_{\WeightsSet[t]}\left(\myVec{s} | \myVec{y}\right). 
\end{equation}
The complexity of applying \eqref{eqn:RxMapping} typically scales linearly with the number of parameters in  $\WeightsSet[t]$~\cite[Ch. 11]{zadeh2015cme}.

\smallskip
{\bf Modular Deep Receivers:}
As 
the number of users can change in time, we consider {\em modular} deep receivers, which are dividable into sub-modules that can be associated with specific users.  Modular architectures  arise when designing deep receivers via model-based deep learning methodologies~\cite{shlezinger2022model}, such as deep unfolding~\cite{balatsoukas2019deep}. One such architecture, used as our main  example, is DeepSIC proposed in~\cite{shlezinger2019deepsic}.

A modular deep \ac{mimo} receiver supporting $K$ users has parameters $\WeightsSet[t]$ that can be divided as $\WeightsSet[t]= \{\Weights^{(k)}[t]\}_{k=1}^K$. Each module $\Weights^{(k)}[t]$ produces the estimate of the symbol of the $k$th user. For instance, DeepSIC is based on soft interference cancellation \ac{mimo} detection, and operates in $Q$ iterations. In each iteration of index $q$, it produces $K$ probability mass functions over $\mySet{S}$, which at the $i$th symbol of the $t$th block is denoted $\{\hat{\myVec{p}}_i^{(k,q)}[t]\}$. These probablity vectors are obtained as soft estimates produced by the \ac{dnn} $\Weights^{(k)}[t]$ applied to the channel output $\myVec{y}_i[t]$ as well as the previous estimates of the interfering symbols, namely,
\begin{equation}
    \label{eq:soft_freq_out}
   \!\! \hat{\myVec{p}}^{(k,q)}_i[t] \!= \!\Big\{P_{\Weights^{(k)}[t]}(\myVec{s}_i^{(k)}[t]|\myVec{y}_i[t], \{ \hat{\myVec{p}}^{(l,q-1)}_i[t] \}_{l\neq k}  )\Big\}_{\myVec{s}\in\mathcal{S}}.
\end{equation}
The output of the receiver are  the soft estimates of the $Q$th iteration, and the estimated conditional distribution is
\begin{equation*}
   P_{\WeightsSet[t]}\left(\myVec{s}_i[t] | \myVec{y}[t]\right) \!=\! \prod_{k=1}^{K[t]} P_{\Weights^{(k)}[t]}\left(\myVec{s}_i^{(k)}[t]|\myVec{y}_i[t], \{ \hat{\myVec{p}}^{(l,Q-1)}_i[t] \}_{l\neq k}  \right).
\end{equation*}


\smallskip
{\bf Training Deep Receivers:} 
Deep receivers, being machine learning models, rely on data to learn the mapping in \eqref{eqn:RxMapping}. Two main learning paradigms are considered in this context: $(i)$ {\em joint learning} and $(ii)$ {\em online learning}. 

{\em Joint training} trains offline using  data corresponding different channel realizations from $\mathcal{H}$.
The \ac{dnn} parameters are static, i.e., they do not change in $t$, and thus must be tuned for a specific network configuration, namely, for a given $K$. 
Training a deep receiver for  $K$ users is done using a data set comprised of channel inputs and outputs corresponding to such networks, given by $\mathcal{D}^{(K)}_{\rm train}=\{(\myVec{y}_i^{\rm train},\myVec{s}^{\rm train}_i)\}_i$. Training the parameters set $\WeightsSet^{(K)}$, dictating the receiver mappings, is done by minimizing the cross-entropy loss, i.e., 
\begin{equation}
\label{eq:joint_training}
    \WeightsSet^{(K)}_{\rm joint} = \arg \min_{\WeightsSet^{(K)}} 
    \mySet{L}_{\rm CE}(\WeightsSet^{(K)},\mathcal{D}^{(K)}_{\rm train}),
\end{equation}
where the cross-entropy loss is defined as
\begin{equation}
\label{eq:cross_entropy_loss}
    \mySet{L}_{\rm CE}(\WeightsSet,\mathcal{D}) = -\frac{1}{|\mySet{D}|}\sum_{(\myVec{y}_i,\myVec{s}_i)\in\mathcal{D}}\log P_{\WeightsSet[t]}\left(\myVec{s}_i | \myVec{y}_i\right).
\end{equation}
Since different architectures are required for different number of users, one has to obtain  parameters sets for each  $K\in[1, K^{\rm max}]$, namely, pre-train and maintain $K_{\max}$  \acp{dnn}. 


{\em  Online training} updates the deep receiver parameters using the pilots  $\mathcal{D}[t]$ on each block $t$. Training again follows  the  cross-entropy loss, but based on the current pilots, namely 
\begin{equation}
\label{eq:online_training}
    \WeightsSet_{\rm online}^{(K[t])}[t] = \arg \min_{\WeightsSet} 
    \mySet{L}_{\rm CE}(\Weights^{(k)},\mathcal{D}[t]).
\end{equation}
Unlike joint learning \eqref{eq:joint_training}, online learning seeks parameters that are suitable for the current channel, using data acquired online. As pilot data is often limited and learning must be done rapidly, one can potentially enhance \ac{sgd} based learning using a principled starting point obtained from previous weights or meta-learned from past channels~\cite{raviv2022online}, requiring  the number of users to be static.


\begin{figure*}
	\centering
	\includegraphics[width = 0.9\textwidth]{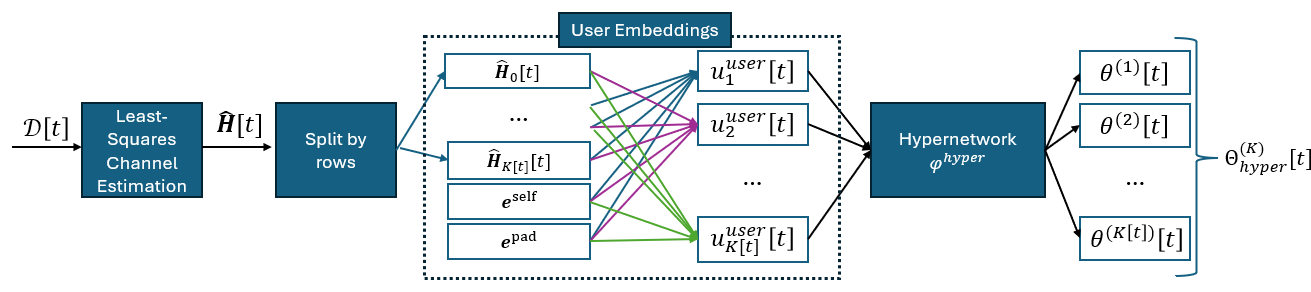}
	\caption{The weights-generation pipeline of modular hypernetworks.}
	\label{fig:hypernetwork}
 \vspace{-0.5cm}
\end{figure*}

\section{Hypernetwork Adaptation of Modular Deep Receivers}
\label{sec:method}
This section presents modular hypernetworks,   combining the adaptivity and flexibility of online learning with the offline training  of joint learning.
We present our hypernetwork framework in Subsection~\ref{subsec:method}, analyze its  online complexity in Subsection~\ref{subsec:complexity}, and provide a discussion in Subsection~\ref{ssec:discussion}.

\subsection{Modular Hypernetworks}
\label{subsec:method}

Our hypernetwork is designed to generate the weights for the deep receiver, accommodating any number of users. As we focus on modular deep receivers, the same hypernetwork can independently generate $\Weights^{(k)}$ for each user $k$, taking into account the channel conditions of the other users. That is, we cast that the weights for each user as a function of the context of the $k$th user, which includes the channel conditions of all other users $\ell\neq k, 1\leq \ell \leq K^{\rm max}$. The resulting procedure, illustrated in Fig.~\ref{fig:hypernetwork}, is comprised of {\em user embeddings}; a {\em hypernetwork} that maps these embeddings into user-wise \ac{dnn} modules; and a dedicated offline training procedure.

\subsubsection{User Embedding}
\label{subsubsec:user_embedder}
Hypernetworks require features of fixed dimensions that are informative of the current context~\cite{ha2016hypernetworks}. Intuitively, the channel parameters $\myVec{H}[t]$ can provide such information. However, as the dimensions of $\myVec{H}[t]$ depend on $K$ which  can vary between blocks, one must seek an alternative formulation.

Accounting for the prospective informativeness of $\myVec{H}[t]$, we construct our user embedding features by first recovering a linear least-squares estimate of the parameters, via  
\begin{equation}
\label{eq:channel_least_squares}
    \hat{\myVec{H}}[t] = \left(\left(\myVec{s}^{\rm pilot}[t]\right)^H   \myVec{s}^{\rm pilot}[t]\right)^{-1} \left(\myVec{s}^{\rm pilot}[t]\right)^H \myVec{y}^{\rm pilot}[t],
\end{equation}
with $(\cdot)^H$ denoting conjugate transpose. While we do not restrict our attention to linear channels, \eqref{eq:channel_least_squares} constitutes a rough first-order estimate of the channel parameters.



The estimate in \eqref{eq:channel_least_squares} is used to embed the physical conditions that each user experiences into a fixed size vector. To that end, we compose the context matrix of each user index $\ell$ on the $k$th user. We differentiate between the active user $\ell=k$, interfering user $\ell\neq k$ and non-existing user $K[t] < \ell \leq K^{\rm max}$, constructing the embedding 
\begin{equation}
\label{eq:user_embedding}
    \boldsymbol{u}_{k}^{(\ell)}[t] = \begin{cases}
      \hat{\myVec{H}}_\ell[t] & \text{if $\ell\neq k$ and $\ell \leq K[t]$}\\
      \myVec{e}^{\rm self} & \text{if $\ell = k$}\\
      \myVec{e}^{\rm pad} & \text{if $K[t]< \ell \leq K^{\rm max}$},
    \end{cases}
\end{equation}
where $\myVec{e}^{\rm self}, \myVec{e}^{\rm pad}$ are $N\times 1$ trainable  vectors corresponding to the user itself and to   non-existing users, respectively. These vectors  are shared for all users' embeddings.

Based on \eqref{eq:user_embedding}, we compose the $N\cdot K^{\rm max} \times 1$ context vector of the $k$th user, denoted $\boldsymbol{u}^{\rm user}_k[t]$,  as the concatenation of 
\begin{equation}
 \boldsymbol{u}^{\rm user}_k[t] = (\boldsymbol{u}_{k}^{(1)}[t] \mathbin\Vert \boldsymbol{u}_{k}^{(2)}[t] \mathbin\Vert \ldots \mathbin\Vert\boldsymbol{u}_{k}^{(K^{\rm max})}[t]),   
 \label{eq:user_embedding_cont}
\end{equation}
with $\mathbin\Vert$ being the concatenation operation. The size of \eqref{eq:user_embedding_cont}  is fixed regardless of the number of instantaneous users $K[t]$.

\subsubsection{Hypernetwork Adaptation}
\label{subsubsec:hypernetwork}
The embeddings in \eqref{eq:user_embedding_cont} are used to generate the parameters of the deep receivers modules associated with each user. 
The hypernetwork is an additional \ac{dnn} with trainable parameters are denoted by $\Network^{\rm hyper}$, whose mapping $\mySet{G}_{\Network^{\rm hyper}}$ transforms each user embedding into its module weights via
\begin{equation}
\label{eqn:hypernetwork}
    \Weights^{(k)}[t] = \mySet{G}_{\Network^{\rm hyper}}\left(\myVec{u}^{\rm user}_k[t]\right).
\end{equation}
The resulting procedure is summarized as Algorithm~\ref{alg:hypernetwork_deepsic}.

  \begin{algorithm} 
    \caption{Modular Hypernetwork Adaptation}
    \label{alg:hypernetwork_deepsic}
    \SetAlgoLined
    \SetKwInOut{Input}{Input}
        \Input{ Pilots set $\mathcal{D}[t]= \{\myVec{s}^{\rm pilot}[t], \myVec{y}^{\rm pilot}[t]\}$; \\
        Information channel output $\myVec{y}^{\rm info}[t]$;\\
        Number of users $K[t]$.} 
    {
    
    Estimate $\hat{\myVec{H}}[t]$ by \eqref{eq:channel_least_squares};

    \For{$k\in\{1,...,K[t]\}$}{
        \For{$\ell\in\{1,...,K^{\rm max}\}$}{ 
                 Calculate $\boldsymbol{u}_{k}^{(\ell)}[t]$ by \eqref{eq:user_embedding};
            }          
                 Concatenate $\boldsymbol{u}^{\rm user}_k[t]$ using \eqref{eq:user_embedding_cont};

                 Get parameters $\Weights^{(k)}[t]$ from $\Network^{\rm hyper}$ via \eqref{eqn:hypernetwork};
    }          

    Set deep receiver $\WeightsSet^{(K[t])}_{\rm hyper}[t] \leftarrow \{\Weights^{(k)}[t]\}_{K[t]}$\\
    \KwRet{ $\hat{\myVec{s}}^{\rm info}[t] = \mySet{F}_{\WeightsSet^{(K[t])}_{\rm hyper}[t]}\left(\myVec{y}^{\rm pilot}[t]\right)$  } \label{line:return}
  }
  \end{algorithm}

We note that the number of weights in each module can generally depend on the number of users $K[t]$. For instance, in DeepSIC each module accounts for the $K[t]-1$ interfering symbols as in \eqref{eq:soft_freq_out}. As the hypernetworks is a \ac{dnn} with a fixed number of output neurons, we set its output size to be  $|\Weights^{(K^{\rm max})}|$. Then, for each $k$th user, we calculate the inferred parameters via \eqref{eqn:hypernetwork}, by taking the $|\Weights^{(K[t])}|$ first outputs as the parameters $\Weights^{(k)}[t]$ of the $k$th user in the $t$th block. We run the  hypernetwork for each context vector $\boldsymbol{u}^{\rm user}_k[t]$, yielding the entire parameters set $\WeightsSet^{(K[t])}_{\rm hyper}[t] = \{\Weights^{(k)}[t]\}_{k=1}^{K[t]}$. 

\subsubsection{Hypernetwork Training}
\label{subsubsec:training_hypernetwork}
The trainable parameters of the  modular hypernetwork are the \ac{dnn} weights $\Network^{\rm hyper}$, and the embedding vectors $\myVec{e}^{\rm self},\myVec{e}^{\rm pad}$. They are designed to be trained offline, as in joint training, i.e., using  the datasets  $\{\mathcal{D}_{\rm train}^{(K)}\}_{K=2}^{K^{\rm max}}$, comprised of multiple  input-output blocks for different values of $K$.

Specifically, for each channel input-output block observed during training, the forward pass follows Algorithm~\ref{alg:hypernetwork_deepsic}. The channel parameters are first approximated  using least-squares estimation \eqref{eq:channel_least_squares}, the user embedding are generated  via \eqref{eq:user_embedding}, and thereafter the parameters $\WeightsSet_{\rm hyper}^{(K)}[t]$  are obtained.
We use conventional deep learning based on \ac{sgd}-based learning, while computing the gradients by  backpropagation through the deep receiver, the hypernetwork \ac{dnn}, and trainable embeddings. The loss that guides the training procedure is the cross-entropy~\eqref{eq:cross_entropy_loss}, namely, it seeks to approach  
\begin{equation*}
    \argmin_{\Network^{\rm hyper},\myVec{e}^{\rm self},\myVec{e}^{\rm pad}} \sum_{K=2}^{K^{\rm max}}\mySet{L}_{\rm CE}\left(\{\mySet{G}_{\Network^{\rm hyper}}\left(\myVec{u}^{\rm user}_k[t]\right)\}_{k=1}^{K},\mathcal{D}^{(K)}_{\rm train}\right),
\end{equation*} 
where the dependence on the embedding vectors $\myVec{e}^{\rm self},\myVec{e}^{\rm pad}$ is encapsulated in the embedding vectors via \eqref{eq:user_embedding}.

\subsection{Complexity Analysis}
\label{subsec:complexity}

Given the considerable resource and latency expenditure associated with re-training, one of the aims of our modular hypernetworks is to support adaptation for the instantaneous channel without the computational burden of online training. To quantify this gain, we analyze the average per-block computational complexity of online learning as compared to our modular hypernetwork approach. 

In our analysis we introduce the following symbols:
\begin{itemize}
    \item {\em Training complexity} $\kappa_{\rm T}(\Weights)$, representing the computational effort in training a \ac{dnn} with parameters $\Weights$ using, e.g., conventional \ac{sgd}-based training.
    \item {\em Inference complexity} $\kappa_{\rm I}(\Weights)$, representing the computational effort of running an inference once through a neural network with parameters $\Weights$. 
\end{itemize}
For \acp{dnn}, both training and inference complexity scale linearly with the number of weights and the data size~\cite[Ch. 11]{zadeh2015cme}. Thus, for a block comprised $B^{\rm pilot}$ pilots and $B^{\rm info}$ information symbols,   we write  
$\kappa_{\rm T}(\Weights) = \alpha_{\rm T}|\Weights|B^{\rm pilot}$ and $\kappa_{\rm I}(\Weights) = \alpha_{\rm I}|\Weights|B^{\rm info}$ for some $\alpha_{\rm T}, \alpha_{\rm I} > 0$ that represent the effort of processing a single symbol using a single parameter in training or inference, respectively. Note that as training involves numerous forward passes and gradient updates, it holds that $\alpha_{\rm T} \gg \alpha_{\rm I}$. 
For simplicity (and also corresponding to the common practice in unfolded architectures \cite{shlezinger2019deepsic, balatsoukas2019deep}), we assume that each module has the same number of parameters, such that $|\Weights^{(k)}| = \frac{1}{K}|\WeightsSet^{(K)}|$ for each $k$. 

Using the above notations, we can characterize the computational savings in terms of average per-block inference complexity of modular hypernetworks compared to online learning, as stated in the following proposition
\begin{proposition}
	\label{pro:online_complexity}
 Consider the transmission of scalar symbols $\mySet{S}\subset \mathbb{C}$ with pilots holding $B^{\rm pilot} > K^{\max}$. Then, the ratio in the average per-block complexity of online learning and modular hypernetwork adaptation when using a modular architecture with $|\Weights|$ parameters per each module with $K$ users satisfies
    \begin{align}
\frac{C_{\rm avg}^{\rm Hyper}}{C_{\rm avg}^{\rm Online}} =& 
\frac{\alpha_{\rm I}(|\Weights|B^{\rm info} +|\Network^{\rm hyper}|) + \mySet{O}(NB^{\rm pilot})}{(\alpha_{\rm T}B^{\rm pilot}+\alpha_{\rm I}B^{\rm info} )|\Weights|}.
	\label{eq:online_complexity}
\end{align}
\end{proposition}
\begin{proof}
    On each block, online learning involves re-training $K$ modules with $|\Weights|$, and inference using a \ac{dnn} with $|\Weights|K$ parameters, hence
        $C_{\rm avg}^{\rm Online} = (\alpha_{\rm T}B^{\rm pilot}+\alpha_{\rm I}B^{\rm info})|\Weights|K$.
    
    The modular hypernetwork infers with the same architecture (at complexity $\alpha_{\rm I}B^{\rm info}|\Weights|K$) without any online training. Instead, it uses  $K$ hypernetwork runs, each with complexity $ \alpha_{\rm I}|\Network^{\rm hyper}|$, and the least squares estimate \eqref{eq:channel_least_squares} from $B^{\rm pilot}$ measurements, at complexity  of $\mySet{O}(NB^{\rm pilot}K)$ as $N\geq K$. 
    Taking the ratio proves the proposition.
\end{proof}



The characterization of the excessive complexity of the considered forms of adaption in Proposition~\ref{pro:online_complexity} accommodates all computations carried out online. One can obtain a more concise (yet faithful) approximation of the complexity-per-block savings assuming the following expected properties: 
\begin{enumerate}[label={P\arabic*}]
    \item \label{itm:traingComp} Training on the pilots is notably more computationally intensive than detection, i.e., $\alpha_{\rm T}B^{\rm pilot} \gg \alpha_{\rm I}B^{\rm info}$.
    \item \label{itm:LSComp} Computing the linear least-squares features \eqref{eq:channel_least_squares} is substantially less complex compared to the subsequent application of the hypernetwork and the \ac{dnn}.
\end{enumerate}
When \ref{itm:traingComp}-\ref{itm:LSComp} hold, the complexity ratio in Proposition~\ref{pro:online_complexity} simplifies to the following corollary:
\begin{corollary}
    \label{cor:computation}
    Under \ref{itm:traingComp}-\ref{itm:LSComp}, Proposition~\ref{pro:online_complexity} implies that
    \begin{equation*}
        \frac{C_{\rm avg}^{\rm Hyper}}{C_{\rm avg}^{\rm Online}} 
        \approx 
         \frac{\alpha_{\rm I}B^{\rm info}  }{\alpha_{\rm T}B^{\rm pilot}} \left(1 + \frac{|\Network^{\rm hyper}|}{|\Weights|B^{\rm info} } \right).
    \end{equation*}
\end{corollary}
Corollary~\ref{cor:computation} reveals that, when the (expected) properties \ref{itm:traingComp}-\ref{itm:LSComp} are satisfied, as long as the hypernetwork is not dramatically more complex compared to the modules it outputs, then $\frac{C_{\rm avg}^{\rm Hyper}}{C_{\rm avg}^{\rm Online}} \ll 1$. Namely, our proposed modular hypernetwork framework is likely to be notably less complex during each block compared to online training. 

\subsection{Discussion}
\label{ssec:discussion}
Both our proposed modular hypernetwork, as well as existing online learning and joint learning approaches, are strategies for handling the dynamic nature of wireless scenarios. In well-known and relatively static test channels, one can train the deep receiver offline and, as long as the training channel remains valid, reliably detect data transmitted over the test channel. However, if the observed channel is dynamic, such as under \ac{mimo} settings where users join and leave the network, each experiencing time-varying conditions, then continuous adaptation of both the architecture and weights of the deep receiver is necessary. Straightforward online training, though, demands a large number of labeled pilots and incurs significant latency and complexity overhead due to the need for training during each coherence duration.

Our proposed hypernetwork-based approach offers a middle ground between the joint and online methods. It presumes access to a dataset from closely related channels, although not identical to the one encountered during testing, and performs training offline. At test time, it requires only a small number of pilots and limited overhead to produce  weights  without any additional training. Thereby, it enables continuous adaptation of the architecture, which elastically matches the current network, as well the weights to match the channel, at lower computational costs than online training.

As modular hypernetworks rely on offline training, they share some of the limitations associated this learning paradigm. For instance, it requires for a large quantity of labeled data from channels similar to that observed on deployment. Being applied in settings that substantially differ from those observed in training  leads to performance degradation, as noted in Section~\ref{sec:Simulation}, although it still outperforms joint learning. Furthermore, our hypernetwork may struggle to scale effectively to scenarios involving a few hundred users due to the exponentially growing space of possible mappings, as dictated by the input and output sizes. This can be possibly tackled by utilizing hypernetworks that do not output the weights directly, but more compact correction terms, see, e.g., context modulation techniques~\cite{ni2024adaptive}. These extensions are left  for future research research.


\section{Numerical Evaluations}
\label{sec:Simulation}

\begin{figure}
    \centering
\includegraphics[width=.4\textwidth]{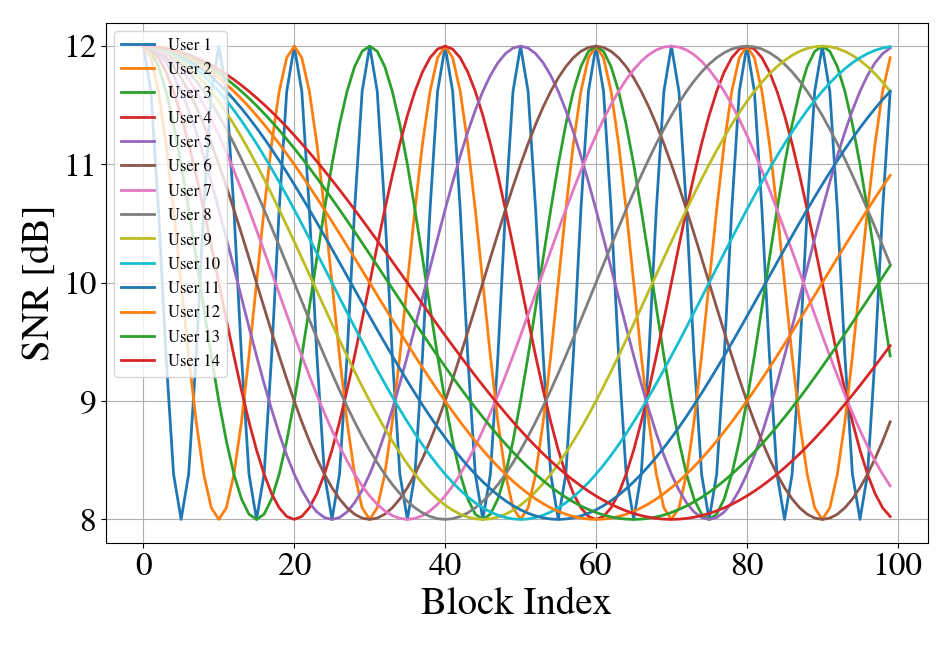}

    \caption{Block-varying \ac{snr} profiles, $K[t]=14$.}
    \label{fig:snr_profiles} 
\end{figure}

\subsection{Experimental Setup}
\label{subsec:compared_approaches}

{\bf Learning Methods:}
We compare the following  schemes: 
$(i)$ {\em Joint Learning}, that trains offline a set of weights for each different users configuration; 
$(ii)$ {\em Online Learning} from pilots;
and $(iii)$ our proposed {\em Modular Hypernetworks}. 
The learning rates are $1\cdot 10^{-3}$ for  joint and online learning, and  $5\cdot 10^{-4}$ for the hypernetwork.  The number of Adam iterations is set as $100$ for the joint and the online methods, and to $25$  for the hypernetwork. This implies that $\alpha_{\rm T} > 100 \alpha_{\rm I}$. 
These values were set empirically to ensure convergence\footnote{The source code used in our experiments is available at \hyperlink{https://github.com/tomerraviv95/adapting-detectors-using-hypernetworks}{https://github.com/tomerraviv95/adapting-detectors-using-hypernetworks}}.

\smallskip
 {\bf Architecture:} We compare the learning methods for training  the modular DeepSIC architecture  \cite{shlezinger2019deepsic}. We unroll $Q=3$ iterations  with user-specific \ac{dnn}-modules comprised of two \ac{fc} layers with sizes $(N + K[t] - 1) \times 16$ and $16 \times 2$ having ReLU activations in-between, thus $|\Weights| = 16(N + K[t] + 1) +18$ (including bias). These values were chosen as a trade-off between its expressiveness and performance. The hypernetwork is composed of two hidden layers of sizes $N\cdot K^{\rm max}\times64$, $64\times 32$ with intermediate ReLU activations, and  a linear output layer of size $16(N + K^{\rm max} - 1)$, corresponding to the parameters of a single   module. The number of hypernetwork parameters is thus 
 $|\Network^{\rm hyper}| = 64 \cdot (N\cdot K^{\rm max}+1) + 32 \cdot 65 + 33 \cdot(16(N + K^{\rm max} + 1) +18)$ (including bias).
\renewcommand{\arraystretch}{1.1} 
\begin{table}[t]
    \centering
      \begin{threeparttable}
        \caption{Runtime over $T=100$ blocks}
          \label{tab:runtimes}
            \begin{tabular}{| c | c |}
                \hline
                \textbf{Learning Method} & \textbf{Average Runtime}\\ 
                \hline
                Joint learning & 2 [s]\\ 
                \hline 
                Online learning & 300 [s] \\
                \hline
                Modular hypernetwork & 3.1 [s] \\
                \hline
            \end{tabular}
      \end{threeparttable}
      \vspace{-0.2cm}
\end{table}

\begin{figure}
    \centering
    \begin{subfigure}[b]{0.48\textwidth}
    \includegraphics[width=\figwidth,height=\figheight]{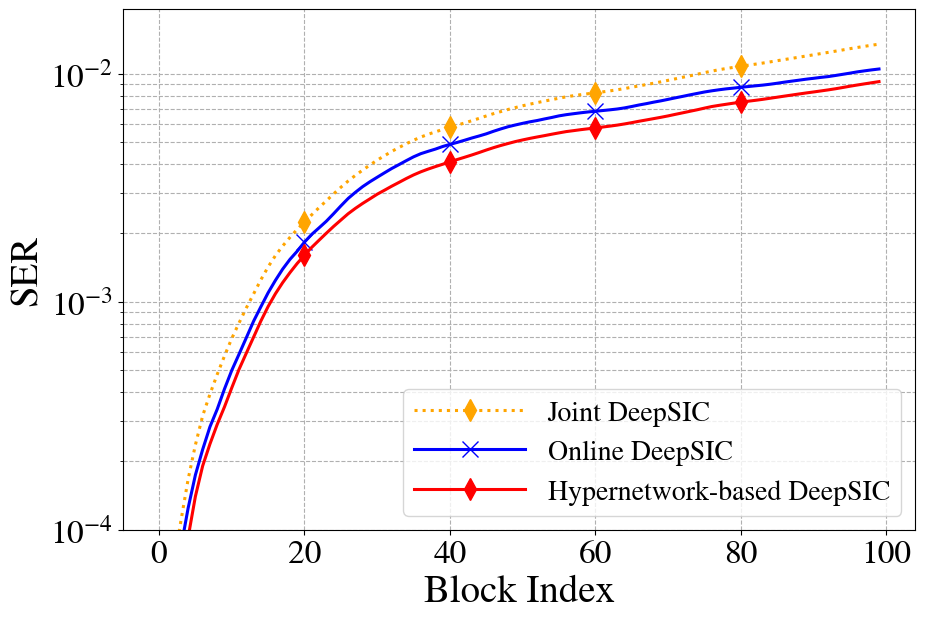}
    \caption{\ac{ser} vs. block index for $K[t]=8$.}
    \end{subfigure}
    \begin{subfigure}[b]{0.48\textwidth}
    \includegraphics[width=\figwidth,height=\figheight]{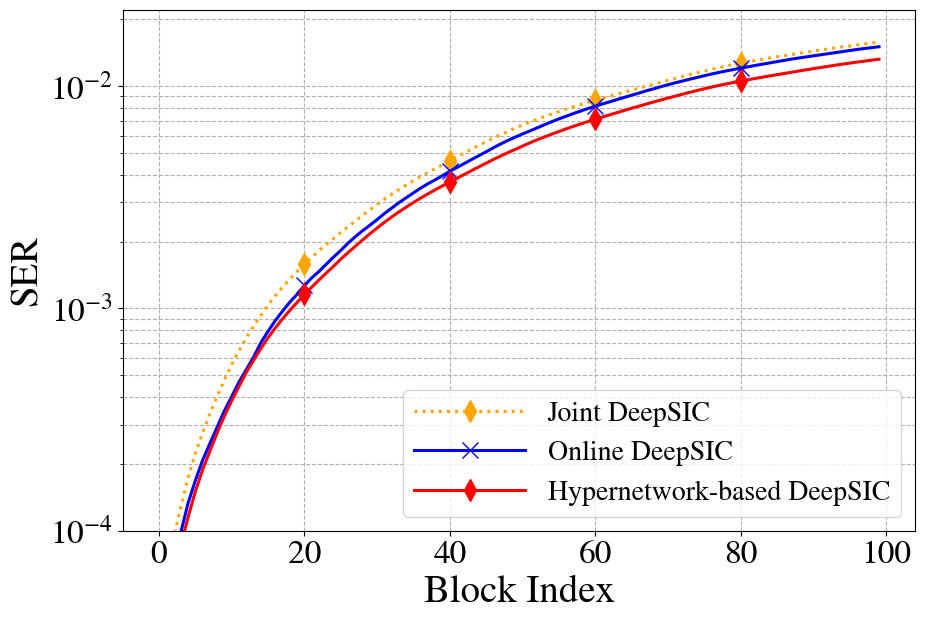}
    \caption{\ac{ser} vs. block index for $K[t]=14$.}
    \label{fig:synthetic_ser_20}
    \end{subfigure}
    \caption{Synthetic channel, time-invariant $K[t]$.}
    \label{fig:weights_synthetic_ser} 
\end{figure}

\smallskip
{\bf Channels:}
We employ a dynamic linear channel model for all evaluations. The input-output relationship of the considered memoryless Gaussian \ac{mimo} channel is 
\begin{equation}
\label{eqn:GaussianMIMO}
\myVec{y}_i[t] = \myVec{H}[t]\myVec{s}_i[t] + \myVec{w}_i[t],
\end{equation}
where $\myVec{H}[t]$ is a $N\times K[t]$ channel matrix, and $\myVec{w}_i[t]$ is white Gaussian noise. We consider two different settings of $\myVec{H}[t]$: 
$(i)$ {\em Synthetic} channels, where  $\myVec{H}[t]$ models spatial exponential received power decay with a different per-user \ac{snr}, and its entries are given by
$\left( \myVec{H}\right)_{n,k} = \sqrt{SNR_{k}} \cdot e^{-|n-k|}$, for each $n \in \{1,\ldots, N\}$ and $ k \in \{1,\ldots, K[t]\}$. The \acp{snr} profiles of each user are varying with the block index (but are constant at each block), as illustrated in Fig.~\ref{fig:snr_profiles}. 
$(ii)$ {\em COST2100} channel, generated from the geometry-based stochastic  model of \cite{liu2012cost}.  The  channel represents multiple users moving at speeds in $0-5$ [m/s] in a semi-urban environment with an operating frequency of $300$[MHz] with mixed urban and rural features. Semi-urban environments often have 
propagation characteristics that are more complex than in purely rural or urban environments, due to the varied types of obstacles and open areas. Succeeding on this scenario requires high adaptivity, since there is considerable variability in the channels observed in different blocks.

The symbols are generated  i.i.d.  from a \ac{bpsk} constellation. Unless stated otherwise, each  of the  $T=100$ blocks is composed of $\Blklen^{\rm pilot}= 800$ pilot symbols, and $\Blklen^{\rm info} = 15,200$ information symbols, i.e., $\Blklen^{\rm tran}= 16,000$ symbols. To obtain the weights for the joint and hypernetwork-based methods, we train offline with $100,000$ symbols per each  $K\in\{2,\ldots,K^{\rm max}\}$.

\begin{figure}
    \centering
    \begin{subfigure}[b]{0.48\textwidth}
    \includegraphics[width=\figwidth,height=\figheight]{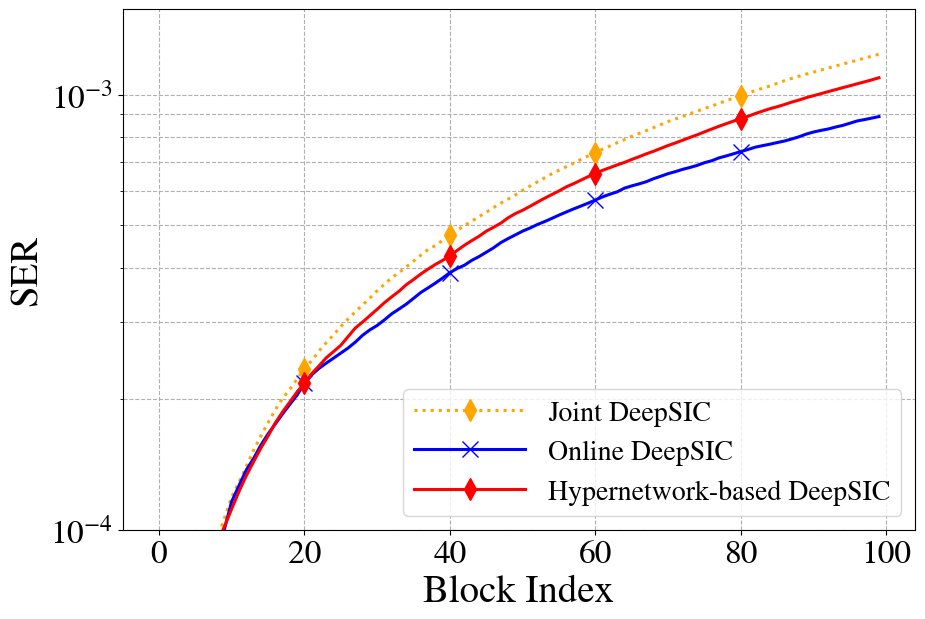}
    \caption{\ac{ser} vs. block index for $K[t]=8$.}
    \end{subfigure}
    \begin{subfigure}[b]{0.48\textwidth}
    \includegraphics[width=\figwidth,height=\figheight]{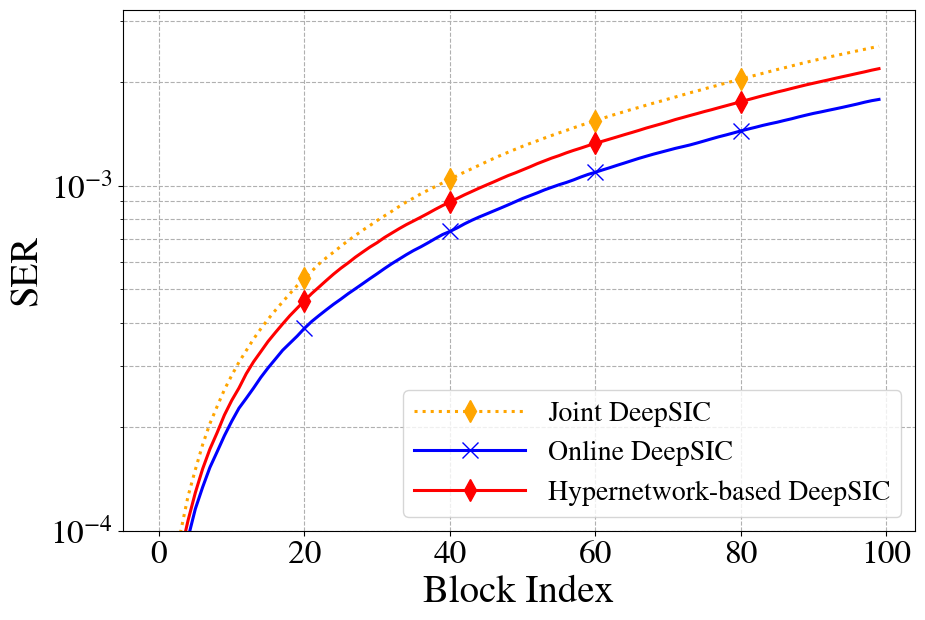}
    \caption{\ac{ser} vs. block index for $K[t]=14$.}
    \end{subfigure}
    \caption{COST 2100 channel, time-invariant $K[t]$.}
    \label{fig:cost_ser} 
\end{figure}

\subsection{Weights Only Adaptation}
\label{subsec:weights}

 We begin by evaluating \ac{ser} for  two configurations that only require weights adaptation, i.e., $K[t]$ remains fixed. In the first, the number of antennas is set to $N=K^{\rm max}=12$ and $K[t]=8$ and in the second $N=K^{\rm max}=20$ and $K[t]=14$.

{\bf Synthetic Channel:}
The aggregated \ac{ser} performance  is reported in Fig.~\ref{fig:weights_synthetic_ser}. We observe that joint learning fails to adapt the receiver for each given channel profile.  Online learning outperforms joint learning by a factor of $\times 2$ continuously, by adapting for each channel realization. Our modular hypernetwork is able to track, and even slightly surpasses online training, while reducing its complexity overhead. 
Specifically, by Corollary \ref{cor:computation}, the complexity saving here holds $\frac{C_{\rm avg}^{\rm Hyper}}{C_{\rm avg}^{\rm Online}}  < {0.2}$, for both configurations. These savings are  translated to even more substantial latency reduction, as reported in Table~\ref{tab:runtimes} (where all runtimes are evaluated on the same RTX3060 GPU). There, we note that the runtime of modular hypernetworks is comparable to  joint learning, and is over $100 \times$ faster than online training. 

\begin{figure}
    \centering
    \begin{subfigure}[b]{0.48\textwidth}
    \includegraphics[width=\figwidth,height=\figheight]{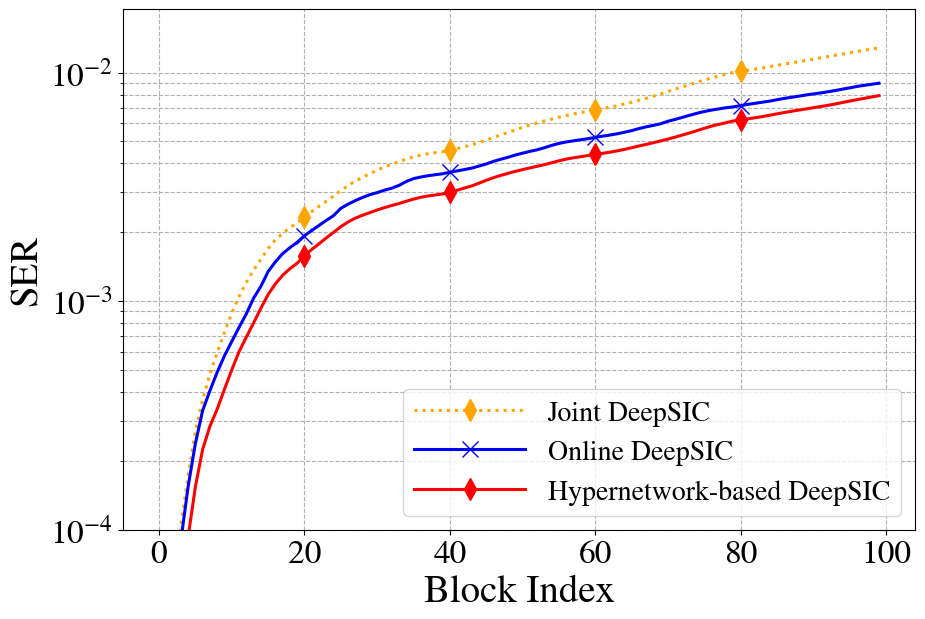}
    \caption{$K[t] \in \{4,5,6,7,8\}$.}
    \end{subfigure}
    \begin{subfigure}[b]{0.48\textwidth}
    \includegraphics[width=\figwidth,height=\figheight]{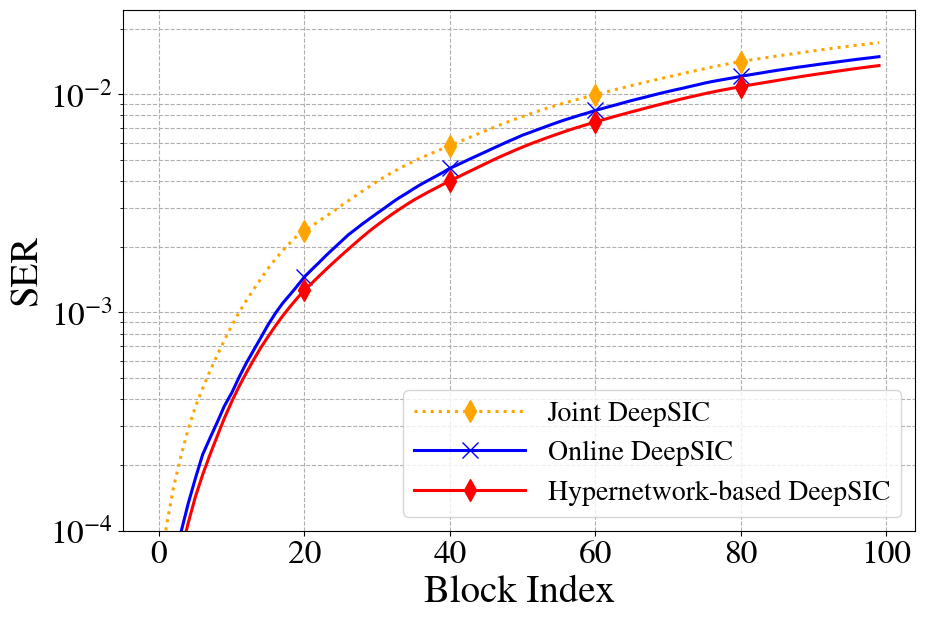}
    \caption{$K[t] \in \{14,15,16,17,18\}$.}
    \end{subfigure}
    \caption{Synthetic channel, time-varying $K[t]$.}
    \label{fig:arch_synthetic_ser} 
\end{figure}

{\bf COST2100 Channel:}
We next consider the  COST 2100 channel with  \ac{snr} of $12$ dB for all users. Fig.~\ref{fig:cost_ser} reports the average \ac{ser} vs. block index over $T=100$ blocks. 
The main conclusion highlighted above is confirmed in this more realistic setting. We systematically observe that modular hypernetworks strike a balance between  joint and online learning, while reducing the overall complexity. 


\subsection{Architecture \& Weights Adaptation}
\label{subsec:architecture_and_weights}

Next, we allow the number of users to randomly change across transmission. Specifically, in the first configuration the number of antennas is set to $N=12$ with $K[t]\in\{4,5,6,7,8\}$, and in the second $N=20$ with $K[t]\in\{14,15,16,17,18\}$. To accommodate different number of users, joint learning has to train offline on all users combinations, resulting in   weights for each $K[t]$, and increasing the memory footprint by an order of $K^{\rm max}$. In contrast, modular hypernetworks maintain a single \ac{dnn}  for a given $K^{\rm max}$.

{\bf Synthetic Channel:}
The average \ac{ser} for the synthetic channel model is depicted in Fig.~\ref{fig:arch_synthetic_ser}. We note that even when not only the weights, but also the architecture itself is required to vary, our modular hypernetwork can still closely match   online learning. For a high number of users, all methods seem to achieve a similar performance. These results are consistent with the ones in Fig.~\ref{fig:weights_synthetic_ser}.

\begin{figure}
    \centering
    \begin{subfigure}[b]{0.48\textwidth}
    \includegraphics[width=\figwidth,height=\figheight]{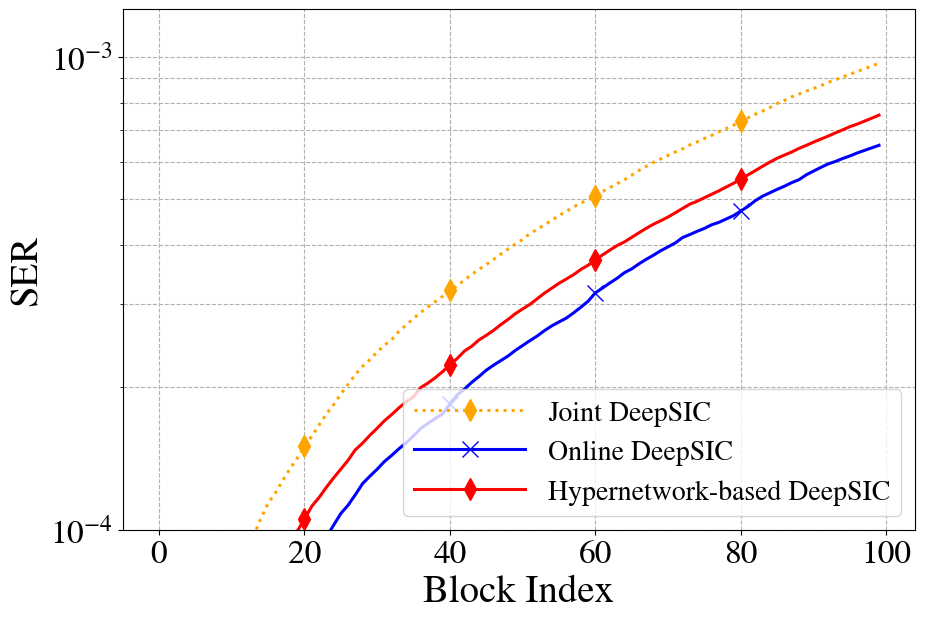}
    \caption{ $K[t] \in \{4,5,6,7,8\}$.}
    \end{subfigure}
    \begin{subfigure}[b]{0.48\textwidth}
    \includegraphics[width=\figwidth,height=\figheight]{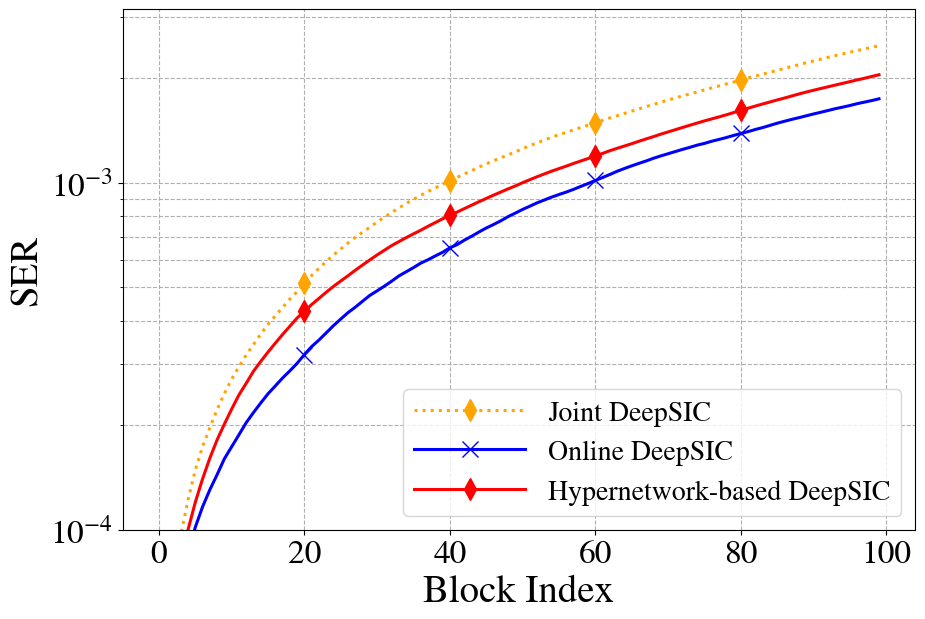}
    \caption{$K[t] \in \{14,15,16,17,18\}$.}
    \end{subfigure}
    \caption{COST2100 channel, time-varying $K[t]$.}
    \label{fig:arch_cost_ser} 
\end{figure}

{\bf COST2100 Channel:}
In Fig.~\ref{fig:arch_cost_ser} we report the \ac{ser} vs. $t$ for COST2100. There, we observe that the performance of the modular hypernetwork is yet better than  joint learning, while improving the complexity (but not the performance) of  online learning  by over $\times{5}$ less computational overhead.

\begin{figure}
    \centering
    \begin{subfigure}[b]{0.48\textwidth}
    \includegraphics[width=\figwidth,height=\figheight]{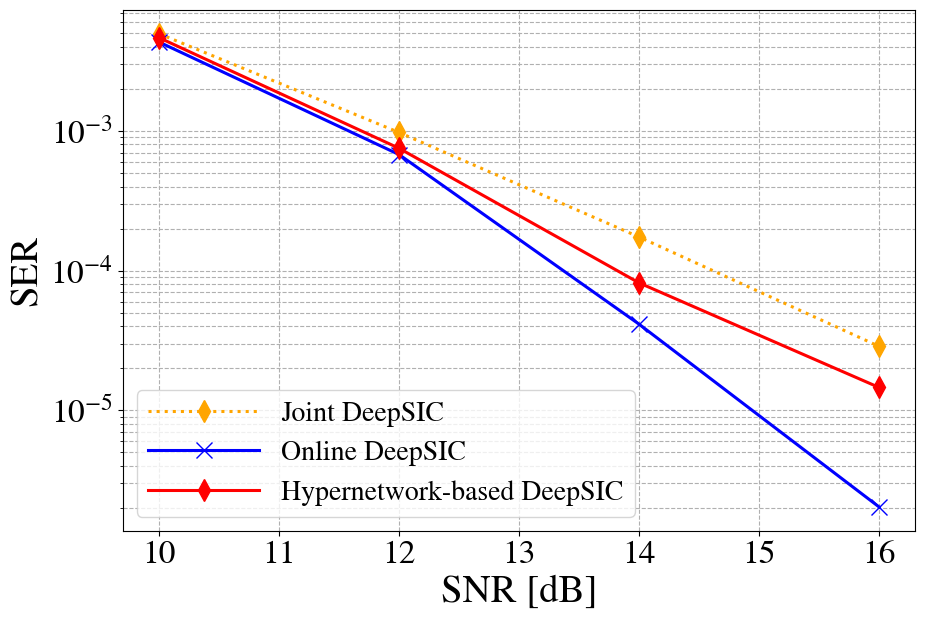}
    \caption{$K[t] \in \{4,5,6,7,8\}$.}
    \label{fig:ser_vs_cost_snr_12}
    \end{subfigure}
    \begin{subfigure}[b]{0.48\textwidth}
    \includegraphics[width=\figwidth,height=\figheight]{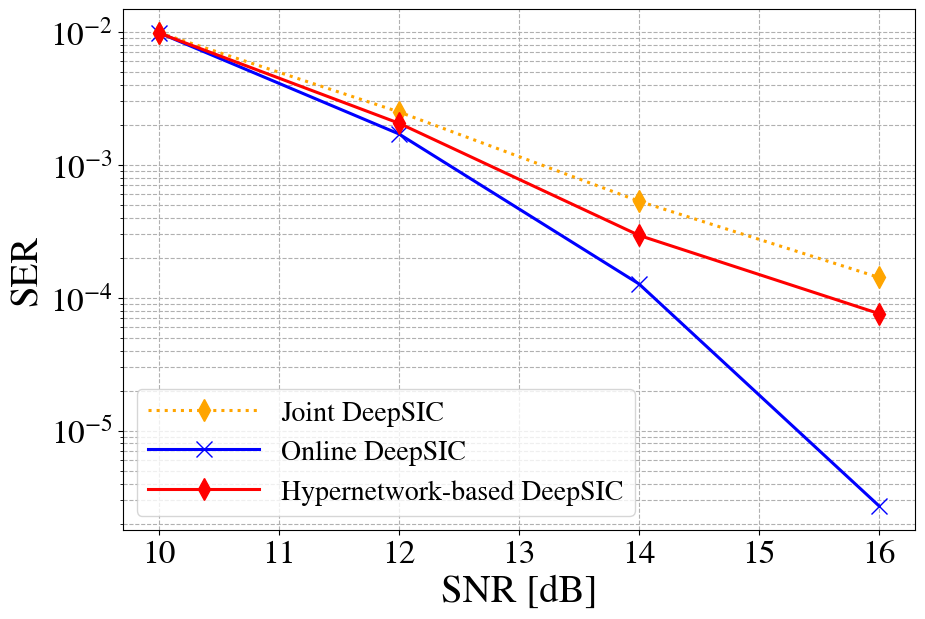}
    \caption{$K[t] \in \{14,15,16,17,18\}$.}
    \end{subfigure}
    \caption{COST 2100 channel, aggregated \ac{ser}.}
    \label{fig:ser_vs_cost_snr} 
    \vspace{-0.2cm}
\end{figure}

To evaluate performance accross a broad range of \acp{snr}, we measure the aggregated average \ac{ser} over the transmission of $T=100$ using different \ac{snr} values (that dictate the variance Gaussian noise vector $\myVec{w}_i[t]$ in \eqref{eqn:GaussianMIMO}). As  observed in Fig.~\ref{fig:ser_vs_cost_snr}, our method matches the \ac{ser} performance of the online training scheme across small to medium \acp{snr}, and is performing better than joint learning in medium-to-high \acp{snr}, yet with smaller complexity and memory footprint. 

\begin{figure}
	\centering
	\includegraphics[width=0.75\columnwidth,height=0.18\textheight]{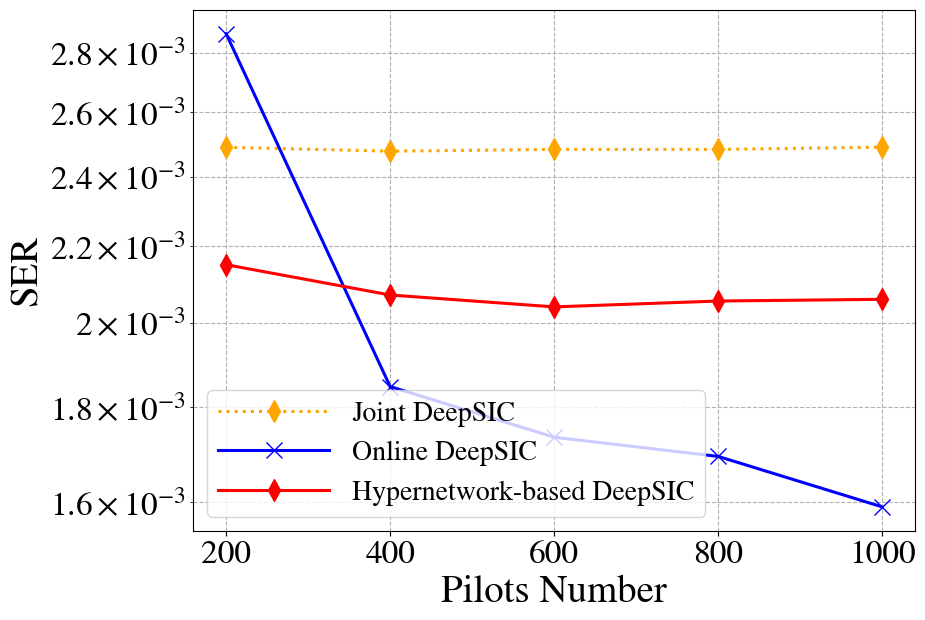}
	\caption{\ac{ser} vs $B^{\rm pilot}$, COST 2100 channel.}
	\label{fig:ser_vs_pilots_number}
\end{figure}

While modular hypernetworks are based on offline training (as in joint learning), their performance also depends on the pilots (as in online learning), from which the features used to set the weights are extracted. Therefore, we next evaluate the effect of varying number of pilots, and robustness to limited pilots. To that aim, we consider the COST 2100 scenario with $12$ dB \ac{snr}, $K[t] \in \{14,15,16,17,18\}$, and $N=20$, employing the same hyperparameters as previously used, but varying the number of pilots each for each $T=100$ blocks transmission.  We observe in Fig.~\ref{fig:ser_vs_pilots_number} that modular hypernetworks can outperform online learning in the small data regime. However, as the available training data grows, our approach does not benefit from additional pilots, similarly to joint learning. Accordingly, the notable complexity reduction of modular hypernetworks comes at the cost of some performance loss compared to online learning.

\section{Conclusion}
\label{sec:conclusion}
We proposed modular hypernetworks  to rapidly adapt deep receivers to channel and network variations. This is achieved by training offline an architecture that generates \ac{dnn} modules online. We provide a complexity analysis, and empirically demonstrate the ability of our approach to provide online adaptation that can approach online training, with similar complexity of pre-trained receivers. 

\bibliographystyle{IEEEtran}  
\bibliography{refs}    

\end{document}